\documentclass[11pt]{article}

\usepackage{amsmath,amssymb,amsfonts,amsthm} 
\usepackage{mathrsfs,latexsym} 
\usepackage{mathtools}

\DeclareMathAlphabet{\pazocal}{OMS}{zplm}{m}{n}

\usepackage[mathscr]{eucal} 

\usepackage{color}
\usepackage{cite}
\usepackage{epsfig}
\usepackage{graphicx}

\usepackage{bm}
\usepackage{cite}

\newtheorem{theorem}{Theorem}[section]

\newtheorem{lemma}[theorem]{Lemma}

\numberwithin{equation}{section}

\newcommand{\RR}{{\mathbb R}}

\newcommand{\TT}{{\mathbb T}}

\newcommand{\Cc}{{\mathcal{C}}}
\newcommand{\Dc}{{\mathcal{D}}}

\newcommand{\Lc}{{\mathcal{L}}}

\newcommand{\Oc}{{\mathcal{O}}}

\newcommand{\Sc}{{\mathcal{S}}}

\newcommand{\bDelta}{\bm \Delta} 
\newcommand{\bnabla}{\bm \nabla}

\newcommand{\bci}{{\bm c}}

\newcommand{\bxi}{{\bm x}}
\newcommand{\byi}{{\bm y}}

\newcommand{\Cf}{{\mathfrak C}}

\newcommand{\Gf}{{\mathfrak G}}
\newcommand{\Hf}{{\mathfrak H}}
\newcommand{\Mf}{{\mathfrak M}}

\newcommand{\Vf}{{\mathfrak V}}
\newcommand{\Wf}{{\mathfrak W}}

\newcommand{\supp}{{\mathrm{supp}}}

\def\ie{{\it i.e.\ }}

\def\be{\begin{equation}}
\def\ee{\end{equation}}  

\begin{document} 

%

\title{\LARGE The universal algebra of the electromagnetic field III. 
              Static charges and emergence of gauge fields}
\author{Detlev Buchholz${}^a$,
Fabio Ciolli${}^b$, Giuseppe Ruzzi${}^b$ and   
Ezio Vasselli${}^b$ \\[20pt]
\small  
${}^a$ Mathematisches Institut, Universit\"at G\"ottingen, \\
\small Bunsenstr.\ 3-5, 37073 G\"ottingen, Germany\\[5pt]
\small
${}^b$ 
Dipartimento di Matematica, Universit\'a di Roma ``Tor Vergata'' \\
\small Via della Ricerca Scientifica 1, 00133 Roma, Italy \\
}
\date{}

\maketitle

\noindent \textbf{Abstract.}   
  A universal C*-algebra of gauge invariant operators is presented,
  describing the electromagnetic field as well as 
  operations creating pairs of static electric charges
  having opposite signs. 
  Making use of Gauss' law, it is shown that the string-localized
  operators, which necessarily connect the
  charges, induce outer automorphisms of the algebra
  of the electromagnetic field. Thus they carry additional
  degrees of freedom which cannot be created by 
  the field. It reveals the fact that gauge invariant operators 
  encode information about the presence of
  non-observable gauge fields underlying the
  theory. Using the Gupta-Bleuler formalism, 
  concrete implementations of the outer automorphisms
  by exponential functions of the
  gauge fields are presented. These fields also appear in
  unitary operators inducing the time translations 
  in the resulting representations of the universal algebra.



\section{Introduction}
\label{sec1}

We construct in this article an extension of the universal C*-algebra
of the electromagnetic field in Minkowski space, presented
in \cite{BuCiRuVa2015}. It includes, 
in addition to the fields, operators creating pairs
of static electric charges with opposite signs. The primary
purpose of our analysis is the 
demonstration that such pairs are inevitably
accompanied by non-observable gauge fields. This generalizes results 
in our previous article \cite{BuCiRuVa2019}, based on a kinematical framework
and relying on canonical commutation relations of the fields.
The upshot of our present investigation is the insight that, quite  
generally, the observable operators encode information about
the existence of non-observable gauge fields, underlying the theory.
In other words, the presence of unobservable gauge fields is traceable by 
physical effects. 

\medskip 
Our restriction to static (infinitely heavy) charges greatly
simplifies the analysis since such charges can sharply be localized; the
localization properties of dynamical charges with finite masses are 
more fuzzy \cite{BuDoMoRoSt}. Yet the appearance of gauge
fields is not related to this idealization. It is solely 
a consequence of Gauss' law, which does not depend on the
mass of a system carrying an electric charge. 

\medskip
Extending arguments in \cite{BuCiRuVa2015}, we will construct
a universal C*-algebra, which in addition to the electromagnetic
field contains operators describing pairs of
opposite charges, localized at spacelike separated points. The
crucial additional input encoded in this algebra is Gauss' law
according to which the values of the individual charges can be
determined by measurements of the electric flux about their
respective positions. The resulting algebra is shown to
describe a local, Poincar\'e invariant net on Minkowski space, describing
the electromagnetic effects induced by static charges. 

\medskip
The principal point in our investigation consists of the demonstration
that the string-localized operators, connecting the charges, define 
outer automorphisms of the subalgebra of the electromagnetic field.
This will be established by different means, firstly in 
the abstract framework and then by exhibiting a representation
of the algebra, based on the Gupta-Bleuler formalism and
an abelian algebra of operators creating static charges. 
Within the present general setting our results corroborate the claim made 
in \cite{BuCiRuVa2015} that gauge fields cannot be replaced
by elaborate limits of the electromagnetic field. 

\medskip 
Our article is organized as follows. In the subsequent section we
establish our notation and recall the definition of the
universal algebra $\Vf$ of the electromagnetic field. 
We will also exhibit outer automorphisms of this algebra
which can be interpreted as \textit{gauge bridges} between static
charges. In Sect.~3 we extend the algebra $\Vf$
to an algebra $\Wf$ containing unitaries,
inducing these automorphisms. For the sake of gauge invariance, 
we add to it an abelian algebra
of unitary operators which generate the static charges.
In Sect.\ 4 we construct physically significant representations of the
electromagnetic part of the resulting 
algebra. The article closes with brief conclusions. 

\section{The universal algebra}
\label{sec2}

In this section we recall the
definition and basic properties of the universal  
C*-algebra $\Vf$ of the electromagnetic field \cite{BuCiRuVa2015},
show how local charge measurements are described,  
and exhibit outer automorphisms which modify the charges.

\medskip 
The unitary elements of the 
algebra, describing exponentials of the 
\textit{intrinsic} vector potential $A_I$, are labeled by  
real, vector-valued test functions $g$ with
compact support and vanishing divergence,
\mbox{$\delta g \doteq \partial_\mu g^\mu = 0$};
they form a real vector space, denoted by $\Cc_1(\RR^4)$.
The relation between $A_I$ and the electromagnetic field
$F$ is given by 
$e^{i a A_I(g)} = e^{i a F(f)}$. Here  $a \in \RR$ and 
$g$ is any solution of the
equation $g = \delta f$ for given real, 
skew-tensor-valued test function $f \in \Dc_2(\RR^4)$,
where $\delta: \Dc_2(\RR^4) \rightarrow \Cc_1(\RR^4)$
is defined by $(\delta f)^\mu \doteq  - 2 \partial_\nu f^{\mu \nu}$.
Such solutions exist according to Poincar\'e's Lemma,
and the potential $A_I$ is unambiguously defined in view of the homogeneous
Maxwell equation, satisfied by the field $F$. 

\medskip
For the sake of mathematical rigor, one describes the formal exponentials
of the intrinsic vector potential by symbols $V(a,g)$ that 
generate a *-algebra~$\Vf_0$ where $a \in \RR$ and
$g \in \Cc_1(\RR^4)$. They  are  
subject to relations, expressing basic algebraic and
locality properties of the potential, given by 
\begin{eqnarray} \label{e.2.1}
& V(a_1,g) V(a_2,g) = V(a_1 + a_2 , g) \, , \ 
V(a,g)^* = V(-a,g) \, , \ 
V(0,g) = 1 \, , \ \ \\ 
\label{e.2.2}
& V(a_1,\delta f_1) V(a_2,\delta f_2) =
V( 1 , a_1 \, \delta f_1 + a_2 \, \delta f_2)
 \ \ \text{if} \ \  
\supp \, f_1 \perp \supp \, f_2 \, , & \\
\label{e.2.3}
& \lfloor V(a_1,g_1) , V(a_2,g_2) \rfloor \, \in \Vf_0 \cap \Vf_0'
\ \ \text{if} \ \ 
\supp \, g_1  \perp  \supp \, g_2 
\, . &
\end{eqnarray}
Here the symbol $\perp$ between two regions indicates that they are
spacelike separated,  $\Vf_0 \cap \Vf_0'$ denotes the center of $\Vf_0$, 
and 
$\lfloor X_1, X_2 \rfloor \doteq X_1 X_2 X_1^* X_2^*$ is the
group theoretic commutator. 

\medskip 
The first relation \eqref{e.2.1} expresses basic properties of
the exponential function and the fact that the
generating operators are unitary. 
Condition \eqref{e.2.2} says that the electromagnetic field
is homogeneous, additive for functions having spacelike separated
supports, and local. Finally, relation \eqref{e.2.3} embodies the information
that the commutator of intrinsic vector potentials, integrated with
spacelike separated test functions, lies in the center of the algebra, 
a fact verified in full generality in \cite{BuCiRuVa2015}. 

\medskip 
As was shown in \cite{BuCiRuVa2015}, there exist faithful states
on the algebra  $\Vf_0$. The corresponding GNS-representations  
determine C*-norms on this algebra. Proceeding to the
completion of $\Vf_0$ with regard to its maximal C*-norm, 
one arrives at a C*-algebra $\Vf$, the universal algebra
of the electromagnetic field. This algebra admits an
automorphic action of the proper orthochronous Poincar\'e
group which is fixed by the relations
\be \label{e.2.4} 
\alpha_P(V(a,g)) \doteq V(a,g_P)  \, , \quad P \in
\Lc_+^\uparrow \ltimes \RR^4 \, , \ g \in \Cc_1(\RR^4) \, ,
\ee
where $x \mapsto g_P{}^\mu(x) \doteq L^\mu{}_\nu \, g^\nu(L^{-1}(x - y))
\in \Cc_1(\RR^4)$ \ for \ $P = (L,y)$. 

\medskip
We shall show now that the algebra $\Vf$ contains all ingredients for the
analysis of electric charges. On one hand, it contains operators which
determine these charges, on the other hand it allows for
the action of automorphisms creating the corresponding fluxes. We recall  
these well-known facts by making use of the electromagnetic field $F$
underlying our framework. In a second step we will
recast the heuristic structures rigorously in terms of the universal algebra. 

\medskip
The electromagnetic field $F$ determines the electric current $J$
by the inhomogeneous Maxwell equation,
\be
J(h) \doteq F(dh) = A_I(\delta d h) \, , \quad h \in \Dc_1(\RR^4) \, .
\ee
Here $\Dc_1(\RR^4)$ is the space of real, vector-valued test functions
with compact support and $dh$ denotes the exterior derivative (the curl)
of $h$. Noticing that $\delta d h \in \Cc_1(\RR^4)$,
the second equality follows from 
the definition of the intrinsic vector potential. 
Choosing a Lorentz frame, the zero component of the current determines
local charge operators for suitable choices of the test functions $h$.
The value of these charges can be changed by maps of the
intrinsic vector potential of the form 
\be
A_I(g) \mapsto A_I(g) + \varphi(g)1 \, , \quad g \in \Cc_1(\RR^4) \, ,
\ee
where $\varphi$ is a real linear
functional on the space $\Cc_1(\RR^4)$. 
Since these maps are compatible with the linear 
properties and causal commutation 
relations of the intrinsic vector potential, they define automorphisms
of the resulting algebra. Applying them to the current, one obtains
\be
J(h) \mapsto J(h) + \varphi(\delta d h) 1 \, , \quad h \in \Dc_1(\RR^4) \, .
\ee
It reveals the fact that the charge can be changed by suitable
choices of~$\varphi$. 

\medskip 
Turning to the details, we fix a Lorentz frame and a corresponding
canonical coordinate system. Given any 
spacetime point $c = (c_0, \bci)$, we consider the
functions $x \mapsto h_c^\mu(x) \doteq \delta^{\mu 0} \,
\tau_{c_0}(x_0) \chi_\bci(\bxi) \in \Dc_1(\RR^4)$. Here
$x_0 \mapsto \tau_{c_0}(x_0)$ is a test function with support in
the interval $[c_0 - \varepsilon, c_0 + \varepsilon]$,  
$\int \! dx_0 \, \tau_{c_0}(x_0) = 1$, and
$\bxi \mapsto \chi_\bci(\bxi)$ is a smooth characteristic function
which is equal to $1$ in a given 3-ball of radius $r$ about $\bci$
and $0$ in the complement of a slightly larger ball with
radius $r + \varepsilon$; the value of $\varepsilon > 0$ may 
be arbitrarily adjusted. The operator $J(h_c)$ describes
a charge measurement in the 3-ball fixed by $\chi_\bci$ at 
the time fixed by $\tau_{c_0}$.

\medskip
For the corresponding intrinsic vector
potential $A_I(\delta d h_c)$, one obtains by a straightforward
computation in the chosen coordinate system
\be \label{e.2.8} 
x \mapsto \delta d h_c(x) =
(\tau_{c_0}(x_0) \bDelta \chi_\bci(\bxi), \, \dot{\tau}_{c_0}(x_0)
\bnabla \chi_\bci(\bxi)) \, .
\ee
Here $\bnabla$ denotes the spatial gradient, $\bDelta$ the Laplacian,
and the dot $\dot{ }$ indicates a time derivative.
Since the spatial derivatives of $\chi_\bci$ vanish in a 3-ball about
$\bci$ and $\tau_{c_0}$ has support around $c_0$, the function 
$\delta d h_c$ has support in a cylindrical surface of
height $2 \varepsilon$ and thickness $\varepsilon$ at distance $r$ 
from $c$.  This feature amounts to
the known fact that charge measurements in a region 
can be replaced by flux measurements at its surface. 

\medskip
In the next step we exhibit 
linear maps $\varphi_m : \Cc_1(\RR^4) \rightarrow \RR$, creating pairs of
opposite charges which are localized at given spacelike separated
regions about points 
$c_1, c_2 \in \RR^4$. To this end we make use of a space of vector-valued
signed measures, which is spanned by vector-valued 
densities $m \in M(\RR^4)$ with compact support,   
satisfying the equation  
\be \label{e.2.9}
\partial_\mu \, m^\mu(x) = 
q \big(\vartheta(c_1 - x) - \vartheta(c_2 - x) \big) \, ,
\quad c_1 \perp c_2 \, ;
\ee
here $\vartheta$ is a non-negative density with support about $0$ which
integrates to $1$ and $\pm q \in \RR$
are the values of the charges carried by the pair.
Basic examples are the measures
with densities given by 
\be \label{e.2.10} 
x \mapsto m^\mu(x) \doteq q \, (c_1 - c_2)^\mu \! \int_0^1 \! du
\, \vartheta(uc_2 + (1-u)c_1 -x) \, , 
\ee
which have support around the spacelike  line connecting
$c_1$ and $c_2$. The maps $\varphi_m$ are now defined as follows.

\medskip \noindent
\textbf{Definition:} Given $m \in M(\RR^4)$, the corresponding
map $\varphi_m$ is defined by
\be \label{e.2.11}
\varphi_m(g) \doteq \int \! dx dy \ m^\mu(x) \, \Delta(x-y) \, g_\mu(y) \, ,
\quad g \in \Cc_1(\RR^4) \, ,
\ee
where $\Delta$ denotes the zero mass Pauli-Jordan commutator function.
These expressions are well-defined since the convolutions of $g$
and $\Delta$ yield smooth functions and the measures fixed by $m$ have compact
support. We shall refer to $\varphi_m$ as pair creating maps. 

\medskip
For the proof that the maps $\varphi_m$ have the desired properties,
we apply them to the functions $h_c$ entering in the definition
of local charge operators. Given $h_c$, let $\Oc(c)$ be the
open double cone
whose basis is a ball of radius $r - \varepsilon$ about $\bci$ in the
time $c_0$-plane; it is the region where charges are measured.
Similarly, let $\overline{\Oc}(c)$ be the closed double cone with
basis of radius $r + 2 \varepsilon$ about $\bci$;
it contains the support of $h_c$.
That the resulting charge operator determines the charges
carried by $\varphi_m$ is apparent from the subsequent lemma. 
\begin{lemma} \label{l.2.1} 
  Let $\varphi_m$ be a pair creating map of charges in sufficiently
  small regions around spacelike separated
  points $c_1, c_2$ and let $h_c \in \Cc_1(\RR^4)$ be a function entering in 
  the definition of local charge operators, as described above. Then 
\be
\varphi_m(\delta d h_c) = 
  \begin{cases}
    q  & \text{if} \quad  c_1 \in \Oc(c), \ c_2 \perp \overline{\Oc}(c)  \\ 
    -q & \text{if} \quad  c_2 \in \Oc(c), \ c_1 \perp  \overline{\Oc}(c)   \\
    0 & \text{if} \quad  c_1, c_2 \in \Oc(c) \ \ \text{or} \ \ 
    c_1, c_2  \perp  \overline{\Oc}(c)  \, . 
  \end{cases}    
\ee
\end{lemma}  
\begin{proof}
  Since the kernel fixed by the Pauli-Jordan function $\Delta$ is a bi-solution
  of the wave equation, one obtains on test functions the equality 
  $\Delta \, \square = 0$, where
  $\square$ is the D'Alembertian. So one can replace the right 
  action of the spatial Laplacian on $\Delta$ by a double time derivative.
  It yields by partial integration 
  \begin{align} \label{e.2.13} 
    \varphi_m(\delta d h_c) & \nonumber
    = - \int \! dx dy \, m^\mu(x) \,
    \dot{\Delta}(x-y) \, \partial_\mu \tau_{c_0}(y_0) \chi_\bci(\byi) \nonumber \\
  & =  \int \! dx dy \, q \big(\vartheta(c_1 - x) - \vartheta(c_2 - x)\big) 
    \, \dot{\Delta}(x-y) \, \tau_{c_0}(y_0) \chi_\bci(\byi)  \, .
  \end{align}
  Now the function
  \be  \label{e.2.14} 
x \mapsto \int \! dy \, \dot{\Delta}(x - y) \tau_{c_0}(y_0) \chi_\bci(\byi)
\ee
is a smooth solution of the wave equation. It vanishes in the spacelike
complement of $\overline{\Oc}(c)$ since
$y \mapsto \tau_{c_0}(y_0) \chi_\bci(\byi)$ has support in that
region. Moreover, in view of the specific choice of the latter function,
the solution is
equal to $1$ in the double cone~$\Oc(c)$, see below. The statement
then follows immediately from the second line in relation \eqref{e.2.13},
provided $\vartheta$ has support in a sufficiently small region
around $0$. 

\medskip
The remaining step is based on standard arguments, which we briefly
recall. One first notices that the replacement of 
$\tau_{c_0}$ by another function $\tau_{c_0}^\prime$ with the 
same properties does not change the value of the
resulting solution within the double cone $\Oc(c)$. This is so because
the integral of $\tau_{c_0} - \tau_{c_0}^\prime$ vanishes, which implies 
$\tau_{c_0} - \tau_{c_0}^\prime = \partial_0 \sigma_{c_0}$, where
$\sigma_{c_0}$ is a test function having also support in the
interval $[c_0 - \varepsilon, c_0 + \varepsilon]$.
Moving the time derivative by a partial integration to $\dot{\Delta}$,
replacing the resulting double derivative by the spatial Laplacian
and moving the latter by partial integrations
to $\chi_\bci$, one obtains for the
difference between the two solutions the function
\be \label{e.2.15}
x \mapsto \int \! dy \, \Delta(x - y) \,
\sigma_{c_0}(y_0) \bDelta \chi_\bci(\byi) \, .
\ee
Now $y \mapsto \sigma_{c_0}(y_0) \bDelta \chi_\bci(\byi)$ has support
in a broadened cylindrical surface of height $2 \varepsilon$ at
spatial distance $r$
from $c$. Taking into account the support properties of~$\Delta$,  
reflecting Huygens' principle, it follows that the  
function \eqref{e.2.15} vanishes in $\Oc(c)$, as claimed. 

\medskip
In order to determine the values in $\Oc(c)$ of the solution
of the wave equation~\eqref{e.2.14}, 
we may now replace the function
$\tau_{c_0}$ by the Dirac delta function at $c_0$. Making use of the
standard properties of $\dot{\Delta}$ in context of 
the Cauchy problem, we note that the resulting solution has, 
at time $c_0$, the value $1$ in a ball of radius
$r$ about $\bci$ and its time derivative vanishes
there. Thus, by the uniqueness properties of solutions of the Cauchy problem,
it is equal to $1$ in the double cone with this base. 
The original solution \eqref{e.2.14} therefore
has the value $1$ in the slightly
smaller double cone $\Oc(c)$, completing the proof.
\end{proof}  

\medskip
We turn now to the algebraic consequences of these 
observations. For the definition of pair creating
automorphisms, we revert to 
the group of unitaries $\Gf_0$, generated by the
symbols $V(a,g)$ with $a \in \RR$, $g \in \Cc_1(\RR^4)$, 
which satisfy relations \eqref{e.2.1} to \eqref{e.2.3}. 
From there we proceed to the 
trivial central extension $\TT \times \Gf_0$ of $\Gf_0$ by the circle
group~$\TT$. It is isomorphic to the unitary group in $\Vf_0$
that consists of the elements $\eta \, V = V \eta$ with $\eta \in \TT$ and
$V \in \Gf_0$. Its generating elements $\eta \, V(a,g)$
satisfy obvious extended equations, 
corresponding to relations \eqref{e.2.1} to \eqref{e.2.3}.

\medskip 
Now, given any real linear
map $\varphi: \Cc_1(\RR^4) \rightarrow \RR$, 
we define a corresponding
map $\beta_\varphi$ on the generating elements, putting 
\be
\beta_{\varphi}(\eta \, V(a,g)) \doteq
\eta \, e^{i a \varphi(g)} \, V(a, g) \, ,
\quad \eta \in \TT \, , \, a \in \RR \, , \, g \in \Cc_1(\RR^4) \, .
\ee
In particular, $\beta_{\varphi}(\eta \, V(a,g)) =
\eta \, \beta_{\varphi}(V(a,g))$ and
$\beta_{\varphi}(V(a,g))^* =  \beta_{\varphi}(V(a,g)^*)$.
Since $\varphi$ is real and linear, 
it follows after a moment's reflection that these maps are
compatible with relations \eqref{e.2.1} to \eqref{e.2.3} and hence
define morphisms, mapping the central extension onto itself. 
Regarding the elements of $\Gf_0$ as basis of
some complex vector space, we extend $\beta_\varphi$
linearly to that space. By the distributive law, we
thereby obtain an automorphism of the *-algebra $\Vf_0$. 
Since $\Vf_0$ has been equipped with the maximal C*-norm, we
conclude that $\beta_\varphi$ extends by continuity
to an automorphism of the C*-algebra $\Vf$.  
One also checks that these automorphisms satisfy 
under the action of Poincar\'e transformations $P$
the equality 
$\alpha_P \beta_\varphi = \beta_{\varphi_P} \alpha_P$,
where $\varphi_P(g) = \varphi(g_{P^{-1}})$ for $g \in \Cc_1(\RR^4)$. 

\medskip 
We identify now the unitary
exponential functions of the local charge operators
with the symbols $V(1,\delta d h_c)$, where $\delta d h_c$ was
given in relation \eqref{e.2.8}. Picking $h_c$ and a pair creating
map $\varphi_m$ as in Lemma \ref{l.2.1}, we obtain for the
action of the corresponding automorphism $\beta_{m} \doteq \beta_{\varphi_m}$ 
on the exponentials 
\be \label{e.2.17} 
\beta_{m}(V(1, \delta d h_c)) = V(1, \delta d h_c) \cdot  
\begin{cases}
  e^{iq} \! \!  \! \! \! & \text{if} \ c_1 \in \Oc(c) \, , \
  c_2 \perp \overline{\Oc}(c) \\
  e^{-iq}  \! \!  \! \! \! & \text{if} \ c_2 \in \Oc(c) \, , \
  c_1 \perp \overline{\Oc}(c)  \\
  1  \! \!  \! \! \! & \text{if} \ c_1, c_2 \in \Oc(c) \ \,  
  \text{or} \  c_1, c_2 \perp \overline{\Oc}(c) \, . 
\end{cases}  
\ee

\medskip 
These relations imply that the maps $\beta_{m}$
define outer automorphisms of~$\Vf$ if $q \neq 0$.
(As a matter of fact, this holds for any 
automorphism complying with the preceding relations.)  
This follows from the existence of a vacuum
representation of $\Vf$, describing the non-interacting electromagnetic
field \cite{BuCiRuVa2015}. There all unitaries
$V(1, \delta d h_c)$, involving the local charge operators,
are represented by~$1$.
It excludes the existence of
unitary operators in $\Vf$ which implement the
action of $\beta_{m}$. In order to be able to
implement it, one must
extend the algebra $\Vf$, as will be discussed in the subsequent
section. 

\section{Extensions of the universal algebra}
\label{sec3}

\noindent The construction of an algebra, 
containing unitary elements implementing
the action of $\beta_{m}$ on $\Vf$, is based
on standard group theoretical arguments. We proceed
from the group $\Hf_0$ that is generated by elements 
$W(m)$,
where $m \in M(\RR^4)$ are real densities with compact support,  
cf.\ equation~\eqref{e.2.10}. These
generating elements are subject to the
relations, $a_1, a_2 \in \RR$, 
\begin{eqnarray} \label{e.3.1}
& W(a_1 m) W(a_2 m) = W((a_1 + a_2) m) \, , \ 
W(m)^* = W(-m) \, , \ 
W(0) = 1 \, , & \quad \\ 
\label{e.3.2}
& W(m_1) W(m_2) = W(m_1 + m_2)  \ \ \text{if} \ \  
\supp \, m_1 \perp \supp \, m_2 \, .  & 
\end{eqnarray}
They encode the information that the symbol $W(m)$ has the
algebraic properties of a unitary exponential function of some
local generator with localization properties determined by
the support of $m$. One can also consistently extend the action of
the Poincar\'e transformations $P$ to the group $\Hf_0$, putting
$\alpha_P(W(m)) = W(m_P)$, where $m_P$ is defined analogously to 
relation \eqref{e.2.4}.

\medskip
We then proceed to the semi-direct product $\Hf_0 \ltimes (\TT \times \Gf_0)$,
putting
\be \label{e.3.3} 
W(m) V = \beta_{m}(V) W(m) \, , \quad W(m) \in \Hf_0 \, , \ 
V \in (\TT \times \Gf_0) \, . 
\ee
The passage from the group of unitaries $\Hf_0 \ltimes (\TT \times \Gf_0)$
to a *-algebra  is now accomplished by standard arguments,
cf.\ \cite{BuCiRuVa2015}. We extend $\Hf_0 \ltimes (\TT \times \Gf_0)$
to a complex vector space $\Wf_0$,
choosing as its basis the elements
$V W$, where $V \in \Gf_0$ and $W \in \Hf_0$. Making use of the
distributive law and relation \eqref{e.3.3},
one obtains an associative product on
$\Wf_0$. Adjoint operators are consistently defined in $\Wf_0$
by canonically
promoting to it the *-operation, acting on the group. In this manner
$\Wf_0$ becomes a *-algebra. One then defines a linear functional
$\omega$ on $\Wf_0$, putting 
\be
\omega(V W) =
\begin{cases}
  0 & \text{if} \quad V W \neq 1 \\
  1 & \text{if} \quad V W = 1 \, .
\end{cases}  
\ee
Thus \ $\omega((V_1 W_1)^* V_2 W_2) =
\omega( \beta(V^*_1 V_2) W_1^*W_2) = 0$
whenever  $V_1 \neq V_2$ or \mbox{$W_1 \neq W_2 \, $};
here $\beta$ denotes the homomorphism
induced by the adjoint action of $W_1$, mapping
$\TT \times \Gf_0$ onto itself. It follows from this
equality that $\omega$ defines a
faithful state on $\Wf_0$. Thus there exists a (maximal) C*-norm on
$\Wf_0$, so by completion
it becomes a C*-algebra~$\Wf$. 

\medskip
Whereas the algebra $\Wf$ contains the desired unitaries,
implementing the automorphisms $\beta_m$ of $\Vf$, its
generating elements may not be
regarded as observables. Because they  admit 
non-trivial gauge transformations given by
\be \label{e.3.5} 
\gamma_s(W(m)) = e^{i \int \! dx \, m^\mu(x) \partial_\mu s(x)} \,  W(m)
= e^{iq ((s \ast \vartheta)(c_2) - (s \ast \vartheta)(c_1)) } \, W(m) \, ,
\ee
where $s$ is an arbitrary real scalar test function
and the asterisk $\ast$ indicates convolution. 
Note that summands in $m$ of elements in 
$\Cc_1(\RR^4)$ do not contribute here. It is also obvious that
correspondent 
gauge transformations act trivially  on the algebra $\Vf$.
By arguments already used, 
it is straightforward to prove that these transformations
define automorphisms of the C*-algebra $\Wf$.

\medskip
In order to obtain gauge invariant unitaries which implement
the action of $\beta_m$, we need to
amend the framework by charged fields, describing the static matter. 
This is accomplished by elements of 
an abelian C*-algebra $\Cf$ that is
generated by all finite sums and products of unitary operators
$\psi(\rho)$, where $\rho$ is a scalar density
on $\RR^4$ whose integral
defines the charge carried by the operator. 
We also assume that 
$\psi(\rho_1) \psi(\rho_2) = \psi(\rho_1 + \rho_2)$,
\ \mbox{$\psi(\rho)^* = \psi(-\rho)$}, and $\psi(0) = 1$.  
Gauge transformations are defined on $\Cf$, putting
on its generating elements 
\be \label{e.3.6}
\gamma_s(\psi(\rho)) = e^{i \int \! dx \, s(x) \rho(x) } \, \psi(\rho) \, . 
\ee
Since $\Cf$, being abelian, is a nuclear C*-algebra, the 
C*-tensor product
\mbox{$\Mf \otimes \Cf$} is uniquely defined. Its subalgebra of gauge
invariant elements,
denoted by $\overline{\Wf \otimes \Cf}$, contains $\Vf$ and for 
any given $m \in M(\RR^4)$ the operators 
\be \label{e.3.7}
\overline{W}(m) \doteq \psi(\vartheta_1) \, W(m) \ \psi(\vartheta_2)^* \, ,
\ee
where $\partial_\mu m^\mu(x) = q (\vartheta(c_1 - x) - \vartheta(c_2 - x))$
and $x \mapsto \vartheta_{1,2}(x) \doteq q \, \vartheta(c_{1,2} - x)$. 
These operators are gauge invariant as a consequence of
relations \eqref{e.3.5} and \eqref{e.3.6}.  
It is also apparent that they satisfy
equation \eqref{e.3.3}, \ie their adjoint action implements
the automorphism $\beta_m$ of $\Vf$ as well. Thus the
C*-algebra~$\overline{\Wf \otimes \Cf}$ contains, apart from the
electromagnetic field and corresponding local charge operators,  also
gauge invariant operators, creating pairs of opposite static charges
and the corresponding fields.

\section{Representations}
\label{sec4}

The algebra $\overline{\Wf \otimes \Cf}$, being a C*-algebra, has
an abundance of representations. Yet since it describes
static charges, it does not have 
representations where energy operators can be defined.
We therefore restrict our attention 
to a subalgebra $\overline{\Wf}$.
It is generated by gauge invariant
operators in $\Wf$ built from elements of $\Vf$ and
operators  ${W}(m)$ with regular densities $m$. 
Such densities are obtained by choosing in
relation~\eqref{e.2.10} test functions $\vartheta$;
we denote the corresponding regular subspace by $\overline{M}(\RR^4)$. 
Operators in $\overline{\Wf}$ of the form 
\be \label{e.4.1} 
V_1 W(m_1) V_2 \cdots W(m_{n}) V_{n+1} \, , \qquad V_1, \dots V_{n+1} \in \Vf \, ,
\ee
are gauge invariant if $\delta \, (m_1 + \cdots + m_n) = 0$.
The regularity properties of $m$ ensure that there exist 
positive energy representations of $\overline{\Wf}$. 

\medskip 
Our arguments are based on the Gupta-Bleuler formalism.
Although we are dealing with gauge invariant
operators, the Gupta-Bleuler fields are needed in order to
obtain concrete representatives of the abstract unitaries $W(m)$. They 
are not separately defined in these representations, but
become meaningful in gauge invariant combinations. 
This fact provides an alternative argument that 
gauge bridges between static charges cannot be 
constructed by means of only the electromagnetic field. 

\medskip
We briefly sketch the well known construction of the Gupta-Bleuler
framework, cf.\ for example \cite{St,Str}.
Let $\Sc(\RR^4)$ be the space of real, vector-valued test functions. We
denote the exponentials of the Gupta-Bleuler fields by the symbols 
\be
e^{iA(u)} \, , \quad u \in \Sc(\RR^4) \, .
\ee
They satisfy the relations, $a \in \RR$, 
\be \label{e.4.3} 
e^{iA(u)} e^{iaA(v)} = e^{i a \langle u, \Delta v \rangle} \, e^{iA(u + a v)} \, ,
\quad \big( e^{iA(u)} \big) {}^* =  e^{- iA(u)} \, ,
\quad  e^{iA(0)} = 1 \, ,
\ee
where we used the short hand notation 
\be
\langle u, \Delta v \rangle \doteq \int \! dx dy \, u^\mu(x) \, 
\Delta(x-y) \, v_\mu(y) \, .
\ee
A linear, hermitian (but not positive), and Poincar\'e
invariant functional $\varpi$ on the
algebra generated by finite sums of these exponentials is given by
\be
\varpi( e^{iA(u)}) =  e ^{\, (1/2) \, \langle u, \Delta_+ u \rangle} \, ,
\quad u \in \Sc(\RR^4) \, .
\ee
Here $\Delta_+$ is the positive frequency part of the Pauli-Jordan
function $\Delta$. The correlation functions
\be
P \mapsto \varpi( e^{iA(u)} \,   e^{iA(v_P)} ) \, ,
\quad P \in \Lc_+^\uparrow \ltimes \RR^4  \, , \ \ u,v \in  \Sc(\RR^4) \, ,
\ee
are continuous and satisfy the relativistic spectrum condition.

\medskip 
One can identify now 
the generating elements of the abstract algebra $\Wf$ with 
exponentials of the Gupta-Bleuler fields, 
\begin{align}
  V(a,g) & \mapsto e^{i a A(g)}  \, , \quad g \in \Cc_1(\RR^4) \, ,
  \ a \in \RR \, , \nonumber \\
  W(m) & \mapsto e^{i A(m)}  \, , \quad m \in \overline{M}(\RR^4) \, , \
  \delta \, m \neq 0 \, .
\end{align}
As is easily checked on the basis of the relations \eqref{e.4.3}, 
this identification complies with all defining algebraic relations 
of $\Wf$. In particular, the automorphisms $\beta_m$,
$m \in \overline{M}(\RR^4)$, 
on $\overline{\Wf}$ are implemented
by non-gauge invariant Gupta-Bleuler operators,
\be
\beta_m(\overline{W}) = e^{iA(m)} \overline{W}  e^{-iA(m)} \, ,
\quad \overline{W} \in \overline{\Wf} \, . 
\ee

\medskip 
The restriction of $\varpi$
to the gauge invariant subalgebra $\overline{\Wf}$ is a
positive functional, describing the vacuum state.
This assertion follows from the fact that this restriction satisfies
the Gupta-Bleuler condition, 
\be
x \mapsto \varpi(e^{iA(u)} \, \delta A(x) \ e^{iA(v)} ) = 0 \quad
\text{if} \quad \delta u = \delta v = 0 \, , \ \ u,v \in \Sc(\RR^4) \, .
\ee
It is a consequence of relations \eqref{e.4.1} and \eqref{e.4.3}. 
Proceeding to the GNS-representation induced by
$\varpi \upharpoonright \overline{\Wf}$, one obtains a continuous,
unitary representation $P \mapsto U(P)$ of the Poincar\'e group,
satisfying the relativistic spectrum condition. Since the
restriction of $\varpi$ to $\overline{\Wf}$
defines a pure state (as a consequence
of its clustering properties), the operators $U(P)$ are
elements of the weak closure of $\overline{\Wf}$. In this sense 
they are gauge invariant.

\medskip
Let us likewise consider the representations of $\overline{\Wf}$
which are induced by the pair creating automorphisms
$\beta_m$, $m \in \overline{M}(\RR^4)$. There arises the question whether
these representations
are also covariant. Indeed, one obtains for the
automorphic action $\overline{W} \mapsto\overline{W}_P$ of the
Poincar\'e group on $\overline{\Wf} $, 
\be
\beta_m(\overline{W}_P) = e^{iA(m)} U(P) e^{-iA(m)} \,
\beta_m(\overline{W}) \
e^{iA(m)} U(P^{-1}) e^{-iA(m)}  \, . 
\ee
But it would be premature to conclude from this observation that
the (formally gauge invariant, unitary)
operators $e^{iA(m)} U(P) e^{-iA(m)}$
are elements of the weak closure of $\overline{\Wf}$, as 
would be necessary for an affirmative answer. 
Note that the representation
$U$ does not induce Poincar\'e transformations of the non-gauge invariant
operators. Otherwise, the cocycles
$P \mapsto e^{iA(m)} e^{-iA(m_P)}$ would be gauge invariant, which
is obviously not the case. This feature contrasts with the case of
covariant localizable charges, where the corresponding cocycles arising from
spacetime translations in charged sectors are generically 
inner in the algebras of observables \cite[Sect.\ II]{DoHaRo}. 

\medskip 
The analysis of these formal representations
of spacetime transformations requires some 
detailed computations, which we briefly sketch. 
The action of the automorphisms $\beta_m$ on the
field operators
$F_{\mu \nu} = (\partial_\mu A_\nu - \partial_\nu A_\mu)$, being
defined in the sense of operator-valued distributions, 
is given by
\be
x \mapsto \beta_m(F_{\mu \nu}(x)) = F_{\mu \nu}(x) +
(\partial_\mu \, \underline{m}_\nu -
\partial_\nu \, \underline{m}_\mu)(x) \, 1 \, .
\ee
Here the zero mass shell restrictions
$\underline{m}$ of the densities $m$ enter,   
\be
x \mapsto \underline{m}_\mu(x) \doteq
\int \! dy \, m_\mu (y) \, \Delta(y - x) \, . 
\ee
Thus the representations induced by $\beta_m$ describe the electromagnetic
field in presence of the classical currents 
$x \mapsto \partial_{\mu \,} \partial^\nu \underline{m}_\nu(x)$. 
The generators of spacetime transformations in the
vacuum representation are spatial integrals of
normal ordered bilinear expressions, involving the electric and magnetic 
field. It follows that the
generators in the representations induced by $\beta_m$ 
coincide with the generators in the vacuum representation, complemented by 
perturbations which are linear in the electromagnetic field,
respectively constant. Since the functions
$\underline{m}$ are smooth and have compact support at
fixed times, these perturbations are well defined and the
perturbed generators are hermitian operators. The proof that
they are also selfadjoint requires some    
thorough analysis. One can show in this manner that 
the unitary time translations are affiliated with 
the weak closure of $\overline{\Wf}$ and have generators
which are bounded from below. We skip these computations
and refer the interested reader to the literature on solutions of
the quantized Maxwell equations involving external currents,
cf.\ for example \cite[Sect.\ V.B.3]{De1}. Underlying 
functional analytic details are discussed in \cite{De2}.

\section{Conclusions}
\label{sec5}

In the present article we have continued our analysis of the universal
algebra of the electromagnetic field by discussing the effects 
of the presence of electric charges. In order to simplify the analysis, 
we have restricted our attention to static (infinitely heavy)
charges. We have also avoided the discussion of infrared problems
by  considering  only neutral pairs of charges which
are localized at finite spatial distances. These restrictions 
allowed us to focus on the impact of the electric charges
on the electromagnetic field. 

\medskip
The upshot of the present investigation is the insight that,
as a consequence of Gauss' law, the
modifications of the electromagnetic field caused by
electric charges  
cannot be described by operations involving only the
electromagnetic field.
Within our framework, it found its expression in the fact that
the gauge bridges between charges
induce outer automorphisms of the universal algebra.
We have therefore enlarged this algebra to a bigger C*-algebra,
containing localized unitary operators implementing these automorphisms.
Noticing that these unitaries are not gauge invariant, we have
added charged field operators,
describing the static matter. Their combination
with the unitaries inducing gauge bridges leads to
well localized gauge invariant
operators, which describe within the C*-algebraic
framework bi-localized static charges.

\medskip 
By making use of the Gupta-Bleuler formalism, we have 
seen that the adjoint action of the
unitary operators creating gauge bridges can be
represented by exponential functions of non-observable gauge fields. 
In this manner gauge fields make their appearance within the framework
of the gauge invariant universal algebra. The adjoint action of these  
non-observable unitaries also leads to physically
meaningful representation of the 
universal algebra, in accordance with the empirical fact that the presence of
electric charges 
has no adverse effects on the energetic
properties of the electromagnetic field. 
Since we had modeled the charged matter as being static it was, 
however, meaningless to discuss its energetic properties as well. 

\medskip
We conclude this article with some remarks on recent related
work by Mund, Rehren and Schroer \cite{MuReSch1, MuReSch2}. These
authors recognized that one may subsume the additional degrees of freedom,  
which are inherent in the gauge bridges, into a scalar, non-local 
field, called escort field. Instead of regarding this escort as  
companion of the electromagnetic field, they propose to add it to
the charged matter part. In this manner they produce gauge
invariant but non-local field operators. Their approach may have
computational
advantages when considering dynamical matter since one can work from the
outset in Hilbert space representations.

\medskip 
The views advocated by these authors 
are not in conflict with the present results. In  
contrast to their approach, we have put forward the localization
properties of physical operations, such as the creation of charged
pairs, without alluding from the outset to the idea of pushing
compensating charges to infinity. Having worked with 
C*-algebras, we have avoided the usage of indefinite
metric formalisms as well. Our excursion to the Gupta-Bleuler formalism
was conducted merely for the sake of illustration. Thus in view of
recent progress in the formulation of dynamical
C*-algebras, cf.\ \cite{BDFR1,BDFR2} and references quoted there,
one may hope that the algebraic framework  
can be expanded into a consistent theory for describing 
the electromagnetic field also in presence of dynamical charged matter. 

\section*{Acknowledgment}

\vspace*{-2mm}
DB  gratefully acknowledges the 
generous and lasting support of Roberto Longo and the
University of Rome ``Tor Vergata'', which made
this collaboration possible. He is also grateful
to Dorothea Bahns and the Mathematics Institute
of the University of G\"ottingen for their continuing hospitality. 
FC and GR acknowledge the MIUR Excellence Department Project awarded to the
Department of Mathematics, University of Rome ``Tor Vergata'', CUP
E83C18000100006, the ERC Advanced Grant 669240 QUEST ``Quantum
Algebraic Structures and Models'' and GNAMPA-INdAM. 
FC is supported in part by MIUR-FARE R16X5RB55W QUEST-NET ``Operator Algebras
and (non)-equilibrium Thermodynamics in Quantum Field Theory''.

\section*{Dedication}

Dedicated to Bert Schroer on the occasion of his 88th birthday.

\end{document}